\documentclass[a4paper,12pt]{article}
\usepackage{localeng}
\ifxetex
\else
\usepackage{mathptmx}
\fi
\usepackage[margin=25mm]{geometry}

\DeclareMathOperator{\KS}{\mathrm{C}\mskip 0.4mu}

\DeclareMathOperator{\E}{\textsf{E}\mskip 0.4mu}
\newcommand{\cnd}{\mskip 1mu | \mskip 1mu}

\newtheorem{theorem}{Theorem}
\newtheorem{proposition}{Proposition}
\newtheorem{lemma}{Lemma}
\theoremstyle{definition}
\newtheorem{definition}{Definition}
\theoremstyle{remark}
\newtheorem{remark}{Remark}
\let\eps=\varepsilon

\title{Inequalities for space-bounded Kolmogorov complexity\thanks{
Authors want to thank the members of the ESCAPE team (especially Ruslan Ishkuvatov), the participants of the Kolmogorov seminar (Moscow) and Algorithmic Randomness workshop for discussions. We are grateful to anonymous reviewers for STACS2021 conference (where the submission was rejected) who suggested many corrections and improvements.}}

\author{ Bruno Bauwens\thanks{National Research University Higher School of Economins, ORCID 0000-0002-6138-0591. Supported by Russian Science Foundation (grant 20-11-20203)}, Peter G\'acs\thanks{Boston University, \texttt{gacs@bu.edu}, ORCID 0000-0003-2496-0332}, Andrei Romashchenko\thanks{\protect\raggedright LIRMM, University of Montpellier, CNRS, Montpellier, France and IITP RAS, Moscow (on leave), \protect\url{https://www.lirmm.fr/~romashchen/}, \texttt{andrei.romashchenko@lirmm.fr}, ORCID 0000-0001-7723-7880. Supported by ANR-15-CE40-0016 RaCAF and RFBR 19-01-00563 grants. Supported by ANR grant FLITTLA},
Alexander Shen\thanks{LIRMM, University of Montpellier, CNRS, Montpellier, France and IITP RAS, Moscow (on leave),  \protect\url{www.lirmm.fr/~ashen}, \texttt{alexander.shen@lirmm.fr}, ORCID 0000-0001-8605-7734. Supported by ANR-15-CE40-0016 RaCAF  and RFBR 19-01-00563 grants. Supported by ANR grant FLITTLA. Part of the work was done while participating in the American Institute of Mathematics Workshop on Algorithmic Randomness (August 2020)}}
\date{}

\begin{document}

\maketitle

\begin{abstract}
  Finding all linear inequalities for entropies remains an important open question in information theory. For a long time the only known inequalities for entropies of tuples of random variables were Shannon (submodularity) inequalities. Only in 1998 Zhang and Yeung~\cite{yeung-zhang} found the first inequality that cannot be represented as a convex combination of Shannon inequalities, and several other non-Shannon inequalities were found soon after that. It turned out that the class of linear inequalities for entropies is rather fundamental, since the same class can be equivalently defined in terms of subgroup sizes or projections of multidimensional sets (Chan, Yeung~\cite{chan, chan-yeung}, Romashchenko, Shen, Vereshchagin~\cite{romashchenko2000}). The non-Shannon inequalities have interesting applications (e.g., to proofs of lower bounds for the information ratio of secret sharing schemes). Still the class of linear inequalities for entropies is not well understood, though some partial results are known (e.g., Mat\'{u}\v{s} has shown in \cite{matus} that this class cannot be generated by a finite family of inequalities).
 
This class also appears in algorithmic information theory: the same linear inequalities are true for Shannon entropies of tuples of random variables and Kolmogorov complexities of tuples of strings (Hammer et al., \cite{romashchenko1997}). This parallelism started with the Kolmogorov--Levin formula~\cite{kolmogorov1968} for the complexity of pairs of strings with logarithmic precision. Longpr\'e  proved in \cite{longpre-thesis} a version of this formula for the space-bounded complexities.

In this paper we prove a stronger version of Longpr\'e's result with a tighter space bound, using Sipser's trick~\cite{sipser}. Then, using this result, we show that \emph{every linear inequality that is true for complexities or entropies, is also true for space-bounded Kolmogorov complexities with a polynomial space overhead}, thus extending the parallelism to the space-bounded algorithmic information theory.
\end{abstract} 
\clearpage

\section{Space-bounded Kolmogorov complexity}

Kolmogorov in his seminal paper of 1965~\cite{kolmogorov1965} defined the complexity of a finite string as the minimal length of a program that produces this string:
\[
\KS_I(x)=\min \{|p|: I(p)=x\}.
\]
Here $I$ is a machine (considered as an interpreter of some programming language), $p$ is a binary string (considered as a program without input), and $|p|$ is its length. In a similar way Kolmogorov defined $\KS_I(y\cnd x)$, the conditional complexity of $y$ given $x$, as the minimal length of a program $p$ that transforms $x$ to $y$: 
\[
\KS_I(y\cnd x)=\min \{|p|: I(p,x)=y\}.
\]
In this case the interpreter $I$ has two arguments (considered as a program and an input for this program).

There exists an interpreter that is optimal to within an additive constant (Solomonoff, Kolmogorov). Different optimal interpreters lead to complexity functions that differ at most by an $O(1)$ additive term. So the complexity can be considered as an intrinsic property of the strings involved. Complexity measures the amount of information in individual finite objects, not random variables (distributions) as Shannon's information theory does. The relation between complexities of strings and entropies of probability distributions is well established: in particular it is shown in~\cite{romashchenko1997} that the same linear inequalities hold for both.

It is easy to see that complexity functions are not computable; moreover, they do not have non-trivial computable lower bounds.
This fact is the basis for Chaitin's famous proof of G\"odel's incompleteness theorem~\cite{chaitin}.

To information theorists, the non-computability of the complexity functions may obscure somewhat their combinatorial significance. A natural approach to the question is to consider versions of Kolmogorov complexity in which the interpreter has some resource (time, space) bounds. This makes the complexity functions computable since now for each program we can run it until it produces some result or exceeds the bound (if the latter does not happen for a long time, we know that there is a loop and the program will never terminate). However, the inequalities for these resource bounded versions become more complex, with different resource bounds on the two sides. In this paper we show a way to overcome these difficulties for the case of space bounds. We will see that each linear inequality holding for entropies holds also for many space-bounded versions of Kolmogorov complexity: the space bound can be chosen from a dense infinite hierarchy of possibilities.

From a more pragmatic point, one could add that unrestricted complexity is not only non-computable, but also irrelevant: if some string has a short program but the time needed to run this program is huge, this string for all practical purposes may be indistinguishable from an incompressible one.

Kolmogorov was aware of these issues: in the last paragraph of~\cite{kolmogorov1965},
he writes that the description complexities introduced above
\begin{quote}
``\ldots have one important disadvantage: They do not take into account the \emph{difficulty} of transforming a program $p$ and an object $x$ into an object~$y$.\footnote{The English translation mentions the ``difficulty of preparing a program $p$ for passing from an object $x$ to an object $y$''. However, the original Russian text is quite clear: Kolmogorov speaks about the complexity of \emph{decompression}, not compression.} Introducing necessary definitions, one can prove some mathematical statements that may be interpreted as the existence of objects that have very short programs, so their complexity is very small, but the reconstruction of an object from the program requires an enormous amount of time. I plan to study elsewhere\footnote{Unfortunately, Kolmogorov did not publish those ``mathematical statements'' about resource-bounded complexity (though he gave some talks on this topic), and his ideas about algorithmic statistics, as the subject is known now, were understood only much later. It turned out that the dependence of $K^t(x)$ on $t$ (if the resource bound $t$ is measured in ``busy beaver units'') gives, for every string $x$, some curve that can be equivalently defined in terms of Kolmogorov structure function, $(\alpha,\beta)$-stochasticity, or two-part descriptions. See \cite{vereshchagin-shen,vereshchagin-shen-long} for a survey of algorithmic statistics.} the dependence of the necessary program complexity $K^t(x)$ on the allowed difficulty $t$ for its transformation into~$x$. Then the complexity $K(x)$ as defined earlier will be the minimum of $K^t(x)$ for unbounded $t$.''\cite[p.~11]{kolmogorov1965}\footnote{Kolmogorov used the notation $K(x)$ for complexity function; now it is usually denoted by $\KS(x)$, while the notation $\mathrm{K}(x)$ is used for the so-called ``prefix complexity''. We follow this convention.}.
\end{quote}

Defining $\KS^r(x)$ and $\KS^r(y\cnd x)$ as the minimal length of the programs that generate $x$ or transform $x$ to $y$ with resource bound $r$, we need to fix some computational model and the exact meaning of the resource bound. It is natural to consider time-bounded or space-bounded computations. The study of time-bounded complexity immediately bumps into the P vs. NP problem~\cite{longpre-thesis,longpre-mocas}, so in this paper we consider only the space-bounded version of complexity. 

Usually the space used by a computation is measured up to a constant factor, but we need more precision. So we should fix a computational model carefully. For Turing machines with arbitrary tape alphabet one should take into account not only the number of cells used but also the alphabet size. If each cell may contain one of $k$ symbols (where $k\ge 2$), then one should multiply the number of used cells by $\log_2 k$. This makes the Turing machine model ``calibrated'' in the following sense: the number of configurations with space not exceeding $s$, is close to $2^s$. In fact, it differs from $2^s$ by a polynomial (in $s$) factor, since we have to take into account the head position (or heads positions for multitape Turing machines). The simulation between models, in our case, the emulation of multitape machines on machines with smaller number of tapes, uses $O(\log s)$ overhead for space $s$ computations, so the space bounds do not depend on the choice of the model up to logarithmic additive terms (this precision is much better than for time bounds).

We need to specify also how the machine gets the input string (strings) and how it produces the output string. If input/output is written on the tape, then the space used by the computation cannot be less than the input/output length. To avoid this artificial restriction, one usually assumes that separate tapes are used for input and output, and make them read- and write-only (so they cannot be used for computations). If we switch from this model to the worktape-only model, we get, in addition to $O(\log s)$, also $O(\text{input/output size})$ space overhead.

For technical reasons, in this paper we use a specific and a bit unusual computation model (finite-state automaton plus two stacks, see below). The results obtained for this model remain valid for multitape Turing machines, but sometimes in a slightly weaker form, namely, with additional $O(\log s)$ terms that appear when we switch from one model to another.\footnote{Strangely enough, it seems that this specific model is essential in our proofs and we do not know how to avoid it even if we agree to have additional $O(\log s)$ terms in our results.} 

Our machines have 
\begin{itemize}
\item one-sided one-directional read-only input tapes with end markers\footnote{For the program tape, the use of the end markers means that we consider plain complexity, not the prefix one (that requires that the interpreter finds by itself where the program ends). However, we allow logarithmic terms in our inequalities for complexities, so the difference between plain and prefix complexity is not important for us.} (one or two tapes, depending on the number of inputs);
\item one-sided one-directional write-only output tape; 
\item two binary stacks (with \textsc{push}/\textsc{pop}/\textsc{empty} requests) as memory devices.
\end{itemize}
All these devices are connected to the finite-state control unit. The tape alphabet is binary, and the sum of the stack lengths is considered as the space measure. 

Note that two stacks are equivalent to a finite tape that can be extended (by inserting an empty cell) or contracted (by deleting a cell) near the head of the Turing machine; the stacks correspond to the parts of the tape on the left and on the right of the head. The head of such a machine knows whether it is at the first/last cell of the tape, and can insert an empty cell or delete a cell before/after the current one.\footnote{To simulate such a tape (with insertions and deletions) on a standard tape, a $O(\log s)$ space overhead is needed: when inserting a cell, we need to move information along the tape to make space for the new cell, and for that we need $O(\log s)$ additional space for counters (or a special marker symbol that makes the alphabet bigger, so the overhead is even worse).}
\begin{center}
\includegraphics[scale=1.05]{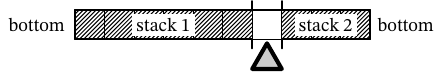}
\end{center}

\begin{definition}
Let $I$ be a machine of the described type with two one-directional read-only input tapes, one one-directional write-only output tape and two binary stacks. We say that $I^s(p,x)=y$ if machine $I$ with inputs $p$ and $x$ produces $y$ and the total length of the two stacks never exceeds $s$ during the computation. We define the conditional space-bounded Kolmogorov complexity as
\[
\KS_I^s(y\cnd x)=\min\{ |p| : I^{s}(p,x)=y\}.
\]
The unconditional version is obtained when the condition $x$ is the empty string.
\end{definition}

Then, as usual, we need a version of the Kolmogorov--Solomonoff universality theorem that says that there exists an optimal machine making the complexity minimal. Now the space bounds should be taken into account, and for our model $O(1)$ additional space is enough:

\begin{proposition}\label{th:optimality-space}
There exists an optimal machine $V$ such that for every machine $P$ there  exists a constant $c$ such that
\[
\KS_V^{s+c}(y\cnd x)\le \KS^s_P(y\cnd x)+c
\]
for all $x,y$.
\end{proposition}

Both machines $P$ and $V$ are of the type we described (we consider only machines of this type if not stated otherwise). We use the same $c$ both for the space overhead and the complexity increase, but this obviously does not matter.

\begin{proof}
Recall a usual construction of a universal machine that can simulate the behavior of an arbitrary machine $M_r$ (an arbitrary finite-state program for the control unit) given its description $r$, for arbitrary inputs $p,x$. Here $r$ is a binary string that describes some machine $M_r$.

We modify this construction to get the required optimal machine $V$. Let us double each bit in $r$ and denote the result by $\overline{r}$. The  machine $V$ is defined in such a way that
\[
V(\overline{r}01p,x)=M_r(p,x).
\] 
  This guarantees that if $r$ is the description of machine $P$ (i.e., $M_r=P$), then
\[
\KS_{V}(y\cnd x)\le \KS_P(y\cnd x)+ 2|r|+2,
\]
since for every $P$-program $p$ we have an equivalent $V$-program $\overline{r}01p$. Therefore, the complexity increase when switching from $P$ to $V$ is $O(1)$. This would be enough if we did not care about the space bounds. But now we need to describe in more details what the simulating machine $V$ does, and check that $V$ uses only $O(1)$ additional space compared to~$M_r$, where $O(1)$-constant may depend on $r$ but not on $p$ and $x$.

  We start with a general remark about our computational model. Let us add an auxiliary tape (a finite read-write tape with insertions/deletions, as explained above) to the two-stack machine we described.\footnote{This is equivalent to adding two more stacks, so we get a machine with four stacks of the same kind. Still in the following lemma the two new stacks are treated differently. Namely, their length is taken into account with some constant factor, so it is more convenient to speak about two stacks and one tape, even if this tape is equivalent to two other stacks.} 

\begin{lemma}\label{lem:sim}
Every machine $M$ of this enhanced type can be simulated by a two-stack machine $M'$ in such a way that at every moment of the computation the space used by $M'$ is bounded by $s_1+O(s_2)$, where $s_1$ is the total length of stacks of $M$, and $s_2$ is the number of cells on the tape at the same moment.
\end{lemma}

\begin{proof}[Proof of Lemma~\ref{lem:sim}]

Let us encode the contents of $M$'s tape in some way (discussed later), and put this encoding between the contents of two stacks of $M$ (reproduced literally, without any encoding). 

\begin{center}
\includegraphics[scale=1.05]{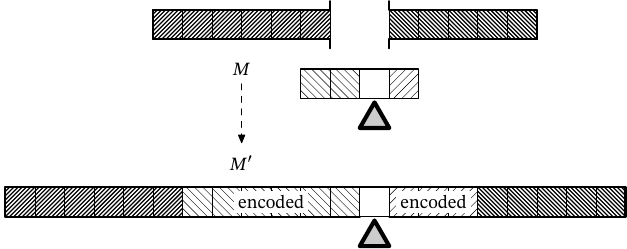}
\end{center}

This will be the contents of the (insertable/deletable) tape of $M'$, and by a \emph{special zone} of this tape we mean the part occupied by the encoding. 

  The head of $M'$ is always in the $O(1)$-neighborhood of the special zone. When $M$ performs an operation on its stack (left or right), $M'$ moves its head to the corresponding endpoint of the special zone and simulates the required operation. When $M$ performs an operation on the tape, $M'$ finds the place in the special zone that corresponds to the head position, and simulates the required operation.

We need an encoding that makes all these operations possible.  For example, we may encode each bit on the tape of $M$ by a group of three identical bits ($000$ or $111$) on the tape of $M'$. Then we use three other $3$-bit blocks (out of $6$ remaining) as left and right endmarkers for the special zone, and as a marker that indicates the position of the $M$-head. The alignment information (position of the $M'$-head modulo $3$) is kept in the finite memory of $M'$. Then $M'$ can distinguish the markers from the encoded bits and find the place it needs (the endpoint or the $M$-head position). 

The space bound for the simulation\footnote{In fact, we may use better encoding and replace $O(s_2)$ by $s_2+O(\log s_2)$, but this is not needed for our purposes.} is easy to check: $M$'s stacks are copied without any overhead, and each bit on $M$'s tape uses $O(1)$ bits in the encoding. We also use $O(1)$ bits for three markers, but this term is absorbed by~$O(s_2)$. 
\end{proof}

Now we describe machine $V$ that uses an additional tape (and then apply Lemma~\ref{lem:sim}). The machine $V$ starts by reading $\overline{r}$ and writing $r$ on its tape, then it skips the separator $01$ and leaves $p$ on the input tape (while $x$ is kept unchanged on the other input tape). Then $V$ executes the program $r$ written on its tape, reading $p$ and $x$ and manipulating the stacks according to $r$'s instructions. For that $V$ needs some additional space on the tape to keep the current state of the simulated program and other information. This space depends only on $r$ but not on $p$ and $x$. It remains to apply Lemma~\ref{lem:sim} to construct an equivalent machine with two stacks; the term $O(s_2)$ appearing in this Lemma depends only on $r$ as required.
\end{proof}

\begin{remark}
  Note that in this argument we used that the input tape is one-directional. Still the result remains valid if we write the input on a bidirectional read-only tape with two endmarkers. In this case we need to distinguish during the simulation whether the input head is inside $p$ or not, but this can be checked by going left by $O(|r|)$ cells and coming back. Note that we have $r$ on the work tape and there is enough space to keep the numbers of size $O(|r|)$.

The other non-standard feature of our model is that it uses two stacks instead of a normal tape. But this feature is not important. We can adapt the argument to standard Turing machines: since the size of the self-delimited block does not change during the simulation, this block can be moved along the normal tape (no cell insertions) without moving the information outside the block (this would happen if the block changed its size on a normal tape).
\end{remark}
 
\begin{remark}\label{remark:timeSimulation}
The same construction works for time bounds (instead of space bounds), but we would get a constant factor instead of an additive constant:
\[
\KS_V^{c\cdot t}(y\cnd x)\le \KS^t_P(y\cnd x)+c,
\]
where $\KS_V^t$ stands for the time-bounded complexity (defined in a similar way). Indeed, each step of a simulated machine now requires several steps of the simulating machine, and the number of these steps is bounded by a constant that depends only on $r$, but not on $p$ and $x$.
\end{remark}
 
Now we fix some optimal machine  $V$, call the corresponding function $\KS_V^s(y\cnd x)$ the \emph{space-bounded Kolmogorov conditional complexity function} and denote it by $\KS^s(y\cnd x)$. The unconditional space-bounded complexity $\KS^s(x)$ can be defined then as $\KS^s(x\cnd \eps)$ for the empty condition $\eps$. It is easy to see that we get an equivalent definition of unconditional complexity if we consider machines $V$ that use only one input tape. Proposition~\ref{th:optimality-space} guarantees that these notions are invariant (do not depend on the choice of the optimal machine) up to $O(1)$ changes in the complexity and the space bound.

\section{Space-bounded complexity of pairs}\label{sec:sym}

The Kolmogorov--Levin theorem (formula for the complexity of pairs, \cite{kolmogorov1968}) says that 
\[
\KS(x,y) = \KS(x)+\KS(y\cnd x)+O(\log\KS(x,y)).
\]
Here $\KS(x,y)$ is the complexity of a pair of strings that is defined as the complexity of some computable encoding for it. For the unbounded complexity the choice of encoding is not important, since any computable transformation changes the unbounded complexity only by an $O(1)$ additive term. For the space-bounded version this is no longer the case, and we define the complexity $\KS^s(x,y)$ as $\KS^s(\overline{x}01y)$, where $\overline{x}$ is $x$ with doubled bits. This encoding of a pair $(x,y)$ treats $x$ and $y$ in different ways, so the natural question is whether the pair complexity as defined above is reasonably robust, e.g., does not change too much when we exchange $x$ and $y$. The following proposition answers this question; note that the space overhead is no more a constant, but is proportional to the size of $x$ and $y$.

\begin{proposition}\label{prop:pairs}
\[
\KS^{s+O(|x|+|y|)}(y,x) \le \KS^s(x,y)+O(1)
\]
for all $s,x,y$.
\end{proposition}

As usual, this means that there exists some $c$ such that $\KS^{s+c(|x|+|y|+1)}(y,x)$ is bounded by $\KS^s(x,y)+c$ for all $s$, $x$, and $y$. We add here $1$ to take care for the special case $|x|=|y|=0$. In this case $x=y$ and the statement is vacuous, but in other statements it could be important. We agree that everywhere the $O(\ldots)$ notation allows $O(1)$ terms, too.

\begin{proof}
  As for the unbounded case, we consider an optimal machine $V(p)$, and then transform it into a machine $\hat V$ that exchanges the pair elements in $V(p)$: if $V(p)=\overline x 01 y$, then $\hat V(p)=\overline y 01 x$. Then we apply Proposition~\ref{th:optimality-space} to the machine $\hat V$. The only thing we need is to be sure that the $\hat V(p)$ computation can be performed in space $s+O(|x|+|y|)$, if $V(p)$ produces $\overline x01y$ in space $s$.

Again we use Lemma~\ref{lem:sim} and equip the machine $\hat V$ with an additional tape; we need only to remember that the space used on this tape is counted with some constant factor. Instead of writing the bits of $\overline x$ on the output tape like $V$ does, the machine $\hat V$ writes them (or just bits of $x$ without duplication) on the work tape, until the next two-bit block is $01$. After that all output bits of $V$ (i.e., bits of $y$) are doubled, so $\overline{y}$ is printed on the output tape. When $V$ terminates, $\hat V$ prints $01$ and after that copies the bits of $x$ from the tape.

It remains to note that our transformation does not change the content of the stacks, and the additional space on the tape is $O(|x|+|y|)$ --- in fact, even $O(|x|)$, since we do not need to store $y$.
\end{proof}

\begin{remark}
A similar argument for standard Turing machines (no insertion of cells allowed) would give additional $O(\log s)$ overhead.
\end{remark}

Now the complexity of pairs is defined, and we would like to develop a space-bounded version of Kolmogorov--Levin formula for the complexity of pairs. This formula says that
\[
\KS(x,y)= \KS(x)+\KS(y\cnd x)+O(\log n),
\]
if $x$ and $y$ are strings of length at most $n$, and contains two inequalities, one in each direction. We want to provide the space-bounded counterparts for them. In one direction this is easy to do. One can get even a bit stronger bound that has term $O(\log \KS^s(x))$ instead of $O(\log n)$; note that for the reasonable values of $s$ we have $\KS^s(x)=O(|x|)$.

\begin{proposition}\label{prop:sym-easy}
\[
\KS^{s+O(|x|+|y|)}(x,y)\le \KS^s(x)+\KS^s(y\cnd x)+O(\log \KS^s(x)).
\]
\end{proposition}

Note the general structure of this statement: we consider an arbitrary bound $s$ on the right-hand side, and on the left side we have to use a slightly bigger bound (for our model $s+O(|x|+|y|)$ is enough).

\begin{proof}
  Let $p$ and $q$ be the minimal programs for $x$ and for $x\mapsto y$. We need to construct the program for the pair $(x,y)$, i.e., for the string $\overline x 01 y$, whose length will be bounded by $|p|+|q|+O(\log |p|)$. This program will work for some other decompressor $\hat V$, and then we use universality to replace $\hat V$ by $V$.

 The program (description) for the pair $(x,y)$ can be constructed as $\overline{l}01pq$ where $l$ is the length of $p$, written in binary. The decoding machine $\hat V$ again uses an auxiliary tape (and then we use Lemma~\ref{lem:sim}). First the machine $\hat V$ copies $\overline{l}01$ to the auxiliary tape. After that the machine $\hat V$ reads and stores $p$ on the tape.  Note that the length of $p$ (i.e., $l$) is already on the tape, so $\hat V$ knows when to stop reading $p$.  Then $\hat V$ simulates the optimal unconditional decompressor on $p$, reading the bits of $p$ from the auxiliary tape and storing the output bits (i.e., bits of $x$) also on the auxiliary tape. Now $\hat V$ is ready to simulate the computation of the optimal conditional decompressor on $q$, reading the bits of $q$ from the input tape (the rest of the input) and using stored bits of $x$ instead of input bits from its second tape. It is easy to see that we need $O(|p|+|q|+|x|)$ cells on the tape (in fact, $O(|p|+|x|)$ cells).

It is not all we need: there is a technical problem. Namely, for small $s$ we cannot guarantee that $|p|+|q|\le O(|x|+|y|)$. So, the space overhead  $O(|p|+|q|+|x|)$) may not be $O(|x|+|y|)$. However, if $|p|+|q|$ significantly exceeds $|x|+|y|$, then we may use the inequality $\KS^{O(|x|+|y|)}(x,y)\le |x|+|y|+O(\log |x|)$ instead. The latter inequality is obtained if we use $x$ and $y$ instead of $p$ and $q$ in the construction above, and use trivial decompressors instead of the optimal ones.
\end{proof}

\begin{remark}
  In the right hand side of Proposition~\ref{prop:sym-easy} we may replace $O(\log\KS^s(x))$ by $O(\log \KS^s(x,y))$. This is not immediately obvious, because we cannot bound $\KS^s(x)$ by $\KS^s(x,y)$ with exactly the same $s$. But for the ``paradoxical'' case $\KS^s(x,y) < \KS^s(x)$ the entire inequality is obviously true.
\end{remark}

The other direction of the Kolmogorov--Levin formula 
is more difficult (both for unbounded and space-bounded complexity).

\begin{theorem}\label{th:sym}
For all strings $x,y$ and for every number $s$ we have
\begin{equation*}\label{eq:sym}
  \KS^{s'}(x)+ \KS^{s'}(y\cnd x)\le \KS^s(x,y)+ O(\log \KS^s(x,y)),
\end{equation*}
  where $s'=s+O(|x|+|y|)$. 
\end{theorem}

Here we use the notation $s'$ for the space bound on the left-hand side to avoid repetitions. The exact meaning of this statement: there exists a constant $c$ such that for all $x,y$ and for every $s$ we have $\KS^{s'}(x)+\KS^{s'}(y\cnd x)\le \KS^s(x,y)+c\log \KS^s(x,y)$, where $s'=s+ c|x|+c|y|+c$. 

This bound assumes that we use the computational model with two stacks; for ordinary Turing machines an additional $O(\log s)$ term is needed in the expression for $s'$.

Longpr\'e~\cite[Theorem 3.13, p.~35]{longpre-thesis} proved essentially\footnote{His setting is slightly different: for him $s$ is not a numerical parameter, but a function of the input size, so the exact comparison is difficult.} the same result with $2s+O(\log s)$ instead of $s$; in his paper he uses $3s$, but his argument gives $2s+O(\log s)$ without changes. We improve this bound using  Sipser's technique from~\cite{sipser} with some additional refinements.

\begin{proof}
  The proof is obtained by a modification of Longpr\'e's argument which in its turn is a modification of the standard proof of the Kolmogorov--Levin formula. So we first recall the standard argument, then explain the modifications used by Longpr\'e, and then prove the final result.

\paragraph*{Recalling the standard argument without resource bounds}
The standard argument (for the complexities without resource bounds) goes as follows. Let $x$ and $y$ be two strings of length at most $n$ and let $m$ be the complexity of the pair: $\KS(x,y)=m$. We may always assume that $m=O(n)$: for the unbounded case it is always true, since we have two strings of length at most $n$ and their pair has complexity at most $2n+O(\log n)$. For the bounded case and very small $s$ the complexity $\KS^s(x,y)$ may be larger than $2n+O(\log n)$ (since even the trivial program for the pair still requires some space to run), but then the inequality is obviously true for $s'=O(n)$, since both terms on the left-hand side are bounded by $n+O(1)$ for this value of $s'$.

Consider the set $S_m$ of all pairs $\langle x',y'\rangle$ such that $\KS(x',y')\le m$. This set can be enumerated by an algorithm (given~$m$), and there are $O(2^m)$ of them. Our pair $\langle x,y\rangle$ is an element of this set. Count the pairs $\langle x,y'\rangle$ in this set \emph{that have the same first coordinate} (i.e., the first coordinate~$x$). Assume that we have about $2^k$ of them for some $k$. We may choose $k$ in such a way that the number of those pairs is between $2^k$ and $2^{k+1}$. Now we make two observations:

\begin{itemize}
\item
Knowing $x$, we can filter the pairs in the enumeration of $S_m$ and keep only the pairs with the first coordinate~$x$, looking at their second coordinate. This process enumerates at most $2^{k+1}$ strings, and $y$ is one of them. The string $y$ can be reconstructed if we know $x$, $m$ and the ordinal number of $y$ in this enumeration (this requires~$k+O(1)$ bits of information). In total we get $O(\log m)+k$ bits (we need to separate $m$ and the ordinal number, and this involves some separation overhead, but this overhead can be absorbed by $O(\log m)$: we may repeat each bit of $m$ twice and add $01$ at the end).  Therefore, $\KS(y\cnd x)\le k+O(\log m)$. 

\item
On the other hand, we can enumerate all $x'$ such that there are at least $2^k$ different $y'$ such that $\KS(x',y')\le m$; there are at most $O(2^{m-k})$ of them, since each of them produces at least $2^k$ pairs and the total number of pairs is $O(2^m)$. The string $x$ appears in this enumeration. So we can specify $x$ by the ordinal number in the enumeration ($m-k+O(1)$ bits), in addition to the values of $m$ and $k$ needed for the enumeration. The total number of bits is $O(\log m)+m-k$, therefore $\KS(x)\le m-k+O(\log m)$. (Note that $k\le m$, so $k$ also has a self-delimited encoding of size $O(\log m)$.) 
\end{itemize}
\noindent
Combining the bounds for $\KS(x)$ and $\KS(y\cnd x)$ and recalling that $m=O(n)$, we get the desired result.
  \smallskip

 \paragraph*{How to obtain a weak space bound (following Longpr\'e)} 

The argument for the unbounded case (as presented above) does not work as is for the space-bounded complexity. The problem is that the enumeration used in this argument needs a lot of space, since the lists of enumerated objects are exponential in $m$. However, another approach is possible. Recall that $x$ and $y$ are strings of length at most $n$. There are at most $O(2^{2n})$ pairs $\langle x,y\rangle$ of strings of length at most $n$. We may consider them in some fixed order (e.g., in the lexicographical one), and compute $\KS^s(x,y)$ for each pair. As we have discussed, the function $\KS^s$ is computable, and the following lemma shows that we do not need too much space to compute it.

\begin{lemma}\label{lem:weakcomp}
The complexity $\KS^s(x)$ can be computed \textup(given $s$ and $x$ such that $s\ge\Omega(|x|)$\textup) in space $2s+O(\log s)+O(|x|)$.
\end{lemma}

This is a weak version of this lemma (that gives Longpr\'e's result). We will see later that one can replace $2s+O(\log s)$ by $s$, and this will allow us to finish the proof of Theorem~\ref{th:sym}, but we start with a simpler bound.

\begin{proof}[Proof of Lemma~\ref{lem:weakcomp}]
  We know that $\KS^{O(|x|)}(x)\le |x|+O(1)$, and our assumption $s= \Omega(|x|)$  guarantees that $\KS^s(x)\le |x|+O(1)$. So it is enough to try all the programs of length at most $|x|+O(1)$ to see which of them produce $x$ with space bound $s$ (in the order of increasing length, so the first one found will be the shortest one). To keep track of the current program, we need $O(|x|)$ space. To simulate the program and to keep track of the space used by it, we need additional $s+O(\log |s|)$ space. The only problem is that the program that we try may never terminate. To detect these cases, we may use a counter for the number of steps. Since a machine with space bound $s$ has at most $2^{s+O(\log s)}$ configurations, if the number of steps exceeds this $2^{s+O(\log s)}$ bound, some configuration appears twice and the program is in the infinite loop. To detect this loop, we use a counter of size $s+O(\log s)$. In total we need $2s+O(\log s)+O(|x|)$ space to find the complexity, as claimed.
\end{proof}

Now the proof goes as before. We consider the set $S^s_{m,n}$ of all pairs $\langle x', y'\rangle$ such that $|x'|\le n$, $|y'|\le n$ and $\KS^s(x',y')\le m$. The pair $\langle x,y\rangle$ is one of its elements. Choose $k$ in such a way that the number of pairs $\langle x,y'\rangle$ in this set (with the first coordinate $x$) is between $2^k$ and $2^{k+1}$. Then
\begin{itemize}
\item $\KS^{2s+O(\log s)+O(n)}(y\cnd x)\le k+O(\log s)+O(\log n)$;
\item $\KS^{2s+O(\log s)+O(n)}(x)\le m-k+O(\log s)+O(\log n)$.
\end{itemize}
Indeed, $y$ can be reconstructed if we know the ordinal number of $y$ in the enumeration of all $y'$ such that $\langle x,y'\rangle\in S^s_{m,n}$, and this set can be enumerated (in the lexicographical order) when $x$, $m$, $n$ and $s$ are known. Lemma~\ref{lem:weakcomp} guarantees that this can be done in space $2s+O(\log s)+O(n)$; recall also that $m=O(n)$ according to our assumption. On the other hand, $x$ can be enumerated together with the other $O(2^{m-k})$ strings $x'$ of length at most $n$ such that there are at least $2^k$ strings $y'$ of length at most $n$ with $\KS^{s}(x',y')\le m$. We can check whether $x'$ has the required property trying all $y'$ sequentially and counting them in $O(n)$ space. The ordinal number of $x$ in the enumeration requires $m-k$ bits, all other parameters require $O(\log n)+O(\log s)$ bits, and the space used in the process is still $2s+O(\log s)+O(n)$.

Combining these two inequalities, we get
\[
\KS^{s'}(y\cnd x)+\KS^{s'}(x)\le \KS^{s}(x,y)+O(\log s)+O(\log n),
\]
where $s'=2s+O(\log s)+O(n)$.

This result is weaker than the claim we need to prove in three aspects. First, we need to replace $2s+O(\log s)$ by $s$ in the expression for $s'$. Second, we proved the inequality with $O(\log s)$ in the right hand side that should not be there. Note that this term makes the statement vacuous if $s$ is exponential in $n$, and does not allow us to get the unbounded Kolmogorov--Levin theorem as a corollary of the bounded version when $s\to\infty$. Finally, we would like to replace $n$ (the length of the strings) in the last term $O(\log n)$ by the complexity of the pair, so the last term would be $O(\log\KS^s(x,y))$.
\smallskip

\paragraph*{How to eliminate factor~$2$ (following Sipser)}

First let us explain how the factor $2$ can be avoided using the following result that goes back to~\cite{sipser}:

\begin{proposition}[Sipser]\label{prop:sipser}
Let $M$ be a machine. Then there is a machine $\overline M$ that decides, given a string $x$ and number~$s$, whether $M$ terminates on input $x$ in space $s$ or not. Machine $\overline M$ uses at most $s+O(|x|)$ space working on pair~$x$,~$s$.
\end{proposition}

In this statement we assume that the two-stack computation model is used; as a corollary, we get the same result with the additional term $O(\log s)$ for other standard models, e.g., multitape Turing machines.

\begin{proof}

We start by proving a weaker statement with a looser space bound $s+O(\log s)+O(|x|)$. For this bound, there is no problem with keeping $x$, the number of input bits already read by~$M$, and the binary representation of~$s$ in the memory. 

We may assume without loss of generality that machine $M$ clears its stacks when terminates, and reads its input completely. For that the old final state is transformed into a cleaning state that pops elements until both stacks are empty, and reads the input until its end. 

Let us consider all configurations of $M$ that use space at most $s$. We include in the configuration the contents of the stacks, the state of the machine, and the position of the input head on $x$ (the input string for which we want to check the termination). We may ignore the output tape: it is write-only, so operations with the output tape do not affect termination. These configurations are considered as vertices of a directed graph. Namely, for every vertex (configuration) $v$ of that kind, draw an edge that goes from $v$ to the next configuration (after one computation step is performed), or no outgoing edges if $v$ is final or if the next computation step violates the space bound. According to our assumption, the final configuration is unique. Let us denote it by $f$. We need to check whether a (unique) path starting from the initial configuration gets into~$f$.

The graph may have cycles (the machine may go into a loop). However, the connected component of the final configuration, i.e., the set of vertices $v$ such that there is a path from $v$ to $f$, is a tree where edges go from a vertex to its parent. Indeed, the outgoing path is unique (the machine is deterministic), so the vertices of any loop cannot have a path to $v$. The root of this tree is $f$. The termination question can now be reformulated as follows: is the initial configuration in the tree?

To answer this question, one can traverse the tree using depth-first search. Note that the standard (non-recursive) algorithm for this (see, e.g., the textbook~\cite[Chapter 3]{shen-progbook}) does not use any additional memory, and the basic operations can be performed with $O(1)$ space. More precisely, let us order siblings (sons of the same parent) arbitrarily (but consistently). This induces a natural ordering on the leaves. We can traverse the tree, visiting all the leaves in this order. In this process we make three types of moves: from a vertex to (a)~its first child, (b)~its parent and (c)~its next sibling (in the chosen order). All non-leaf vertices are visited twice: first on the path from the root to leaves, the second time on the way back to the root.  The tree-traversing algorithm at every step keeps the current vertex and one bit that says whether we are on the way to the leaves or back.  The basic operation of the tree-traversing algorithm are the following:
\begin{itemize}
\item \emph{Checking whether the given vertex $v$ has children, and if yes, finding the first child of $v$}. In our case this means that the current configuration can be obtained from some other configuration; if yes, we should find the first among those predecessor configurations (children).

\item \emph{Checking whether the given vertex $v$ is the last sibling in the ordering we have on $v$'s siblings; finding the next sibling of $v$ if it exists}. In our case we should consider all the configurations that have the same successor, and find the next one in the chosen ordering (if our configuration is not the last one).

\item \emph{Checking whether the given vertex $v$ is the root, and finding the parent of $v$ if $v$ is not the root}. In our terms it means that we have to check whether the configuration is final, and find the successor configuration if it is not.
\end{itemize}

We also need to keep track of the configuration size (since we do not consider configurations that require more than $s$ space), but this can be done in $O(\log s)$ memory. All other checks are local (require $O(1)$ additional memory), since only the immediate neighborhood of the head ($O(1)$ top elements of the stacks) needs to be taken into account.

We need also to keep track of the position of the input head in $x$ (and keep $x$ in the memory), but this is easy to do with $O(|x|)$ overhead. This finishes the argument for $s+O(\log s)+O(|x|)$ bound.
\smallskip

To get rid of $O(\log s)$ in this bound (as promised), we need additional (and rather strange) tricks. The machine $M$ has two stacks, as well as the machine $\overline{M}$ that we need to construct. However, it is convenient to use Lemma~\ref{lem:sim} and add an auxiliary tape to $\overline{M}$; the space used on this tape is taken into account with some constant factor.

  We keep $x$ (and the input head position in $x$) on the auxiliary tape; this requires $O(|x|)$ space and is not a problem. We use the stacks of $\overline{M}$ to keep (literally) the contents of $M$'s stacks in the current position (i.e., the current vertex considered by the tree traversal algorithm). The basic operations listed above are local and do not require memory (except for $x$ and the input position, already taken into account). However, we need to check whether the modified position of $M$ still uses space at most $s$, i.e., that the total size of two stacks still does not exceeds $s$ after possible increase in the stack sizes. Before, having $O(\log s)$ additional space, we could keep the value of $s$ and the current lengths of stacks, and make these checks. What can we do now? The following idea helps: let us remember $s$ all the time, but in an indirect way: we keep on the auxiliary tape the \emph{difference between $s$ and the total length of two stacks} (of $M$ or $M'$, they are the same). This difference is enough to check whether the possible neighbor in the tree is valid (has total stack length at most $s$). When the total length approaches $s$, the difference counter is small and requires only $O(1)$ bits. When stacks are short, the counter is big and may require $O(\log s)$ bits --- but since we measure the total length of the stacks and the tape, these $O(\log s)$ additional bits are not a problem (it is easy to see that $k+O(\log (s-k)) \le s+O(1)$ for all $k< s$). This finishes the proof of the Proposition~\ref{prop:sipser} (in its strong form, without $O(\log s)$ term).
\end{proof}

Sipser's trick allows us to prove the following stronger version of Lemma~\ref{lem:weakcomp}:

\begin{lemma}\label{lem:strongcomp}
The complexity $\KS^s(x)$ can be computed \textup(given $s$ and $x$ such that $s\ge\Omega(|x|)$\textup) in space $s+O(|x|)$.
\end{lemma}

\begin{proof}
  In the proof of Lemma~\ref{lem:weakcomp} we need to keep $s$ and test all the possible programs to check whether they produce $x$ within space bound $s$. For that, we first check that a program terminates in space $s$ using Proposition~\ref{prop:sipser}, and if yes, apply the interpreter to the program (now being sure that we do not violate the space bound) and compare the output with $x$. Again, we can keep $s$ indirectly during both phases, by keeping the difference between $s$ and the total length of the stacks. Then, if this length comes close to $s$, the counter is small, and when the stacks are short, we may use a lot of space for the counter.
\end{proof}

This immediately gives us a better bound: the inequality
\[
\KS^{s'}(y\cnd x)+\KS^{s'}(x)\le \KS^{s}(x,y)+O(\log s)+O(\log n)
\]
is now proven for $s'=s+O(n)$, (now we have $s$ instead of $2s+O(\log s)$). The $O(\log s)$ additive term \emph{in the right hand side} is still there. It was used to remember the space bound $s$, so the configurations of size greater than $s$ could be discarded.  Still we can avoid this $O(\log s)$ term if we change the enumeration order. 

\paragraph*{Eliminating $O(\log s)$ term in the right hand side}

Instead of enumerating for some fixed $s$ all pairs $\langle x',y'\rangle$ with $|x'|,|y'|\le n$ such that $\KS^s(x',y')\le m$, we enumerate the pairs such that $\KS^u(x',y')\le m$ (and $|x'|,|y'|\le n$) sequentially for $u=1,2,3,\ldots$ So every pair with $|x'|,|y'|\le n$ and (unbounded) complexity $\KS(x',y')\le m$ will be enumerated at some stage, but the space bound (and also the amount of space used for the enumeration) increases with time. Using some additional precautions, we may guarantee that this enumeration will be without repetitions (no pair is enumerated twice). Indeed, after we find some $x',y'$ with $\KS^u(x',y')\le m$ for the current $u$, we check whether the same is true for bound $u-1$, and if yes, skip the pair.  Note that it can be done without increasing the space usage, since we may reuse the same space for both bounds $u$ and $u-1$.

In other words, we enumerate all the pairs $\langle x',y'\rangle$ with $|x'|,|y'|\le n$ and $\KS(x',y')\le m$ in the following order: we compare the minimal space $u$ needed to establish the inequality $\KS^{u}(x',y')\le m$, and for the same $u$ we use some standard ordering on pairs. In this way every pair is enumerated only once without the need to keep the list of the pairs already enumerated.

Now we choose the value of $k$, like we did in the in the proof of the Kolmogorov--Levin formula for unbounded complexity. For that we consider the pairs such that $\KS^s(x,y')\le m$ for given~$x$ and arbitrary~$y'$ (such that $|y'|\le n$) \emph{for the given value of $s$}. There exists some $k$ such that the number of these pairs is between $2^k$ and $2^{k+1}$.  

Let us check that $\KS^{s'}(y\cnd x)\le k+O(\log n)$ for $s'=s+O(n)$. Knowing $n$ and $m$, we can perform the enumeration described above; knowing $x$, we can restrict the enumeration to pairs $\langle x,y'\rangle$ with the first component $x$. The pair $\langle x,y\rangle$ is among them; moreover, we know that its ordinal number in this restricted enumeration is at most $2^{k+1}$, so we need $k+O(1)$ bits to specify this number (in addition to $n$, $m$ and $x$). Performing the enumeration until that many pairs with first component~$x$ appear, we use only $s+O(n)$ bits, since we stop the enumeration after the required number of pairs are found.  This gives the inequality we wanted (recall that $m=O(n)$, so we can specify $m$ and $n$ by $O(\log n)$ bits).

Now we need to show that $\KS^{s'}(x)\le m-k+O(\log n)$ for the same value of $s'$. For that we enumerate elements $x'$ that have large ``vertical sections'' (have many $y'$ such that $\KS^{s}(x',y')\le m$). Again we do it sequentially for $u=1,2,3,\ldots$ For each $u$ we run a loop over all $x'$ with $|x'|\le n$. For each of them we count all $y'$ such that $|y'|\le n$ and $\KS^u(x',y')\le m$. This is done sequentially (and we reuse the space at every step).  If there are more than $2^k$ different strings $y'$ found, we include $x'$ in the enumeration of the elements that have large vertical sections. To avoid repetitions, we use the same trick: we check whether the size of the vertical section was not large enough for the previous value of $u$, repeating all the computations with this value. In this way we enumerate all $x'$ that have large sections for unbounded complexity, using more and more space in the process. Note that we keep the current value of $u$ all the time, but indirectly, as a combination of current stacks' length and the counter (and the counter is short when the space is tight).

This enumeration will include our $x$ at some stage $u\le s$. The ordinal number of $x$ in the enumeration is at most $2^{m-k+O(1)}$ for the same reason as before (for complexities without space bounds).  At this stage the space used by the computation is $s+O(n)$, and we stop the enumeration after a required number of strings are enumerated. To start the enumeration we need to know $n$, $m$ and $k$, all three can be specified by $O(\log n)$ bits, in total we get $m-k+O(\log n)$ bits. We see that $\KS^{s'}(x)\le m-k+O(\log n)$ and may combine this inequality with the bound for $\KS^{s'}(y\cnd x)$ obtained earlier, thus eliminating the term $O(\log s)$ in the right hand side as promised.

\paragraph*{Replacing $O(\log n)$ by $O(\log\KS^s(x,y))$}
We have proven the inequality 
\begin{equation}\tag{$*$}\label{eq:logn_precision}
\KS^{s'}(x)+\KS^{s'}(y\cnd x)\le \KS^s(x,y)+O(\log (|x|+|y|))
\end{equation}
for $s'=s+O(|x|+|y|)$, (in our notation $n$ was the maximal of the lengths of $|x|$ and $|y|$). The last step is to replace $|x|+|y|$ in the right hand side by $\KS^s(x,y)$. This means that in our argument we do not have $n$ as a parameter of the enumeration process and may use only $m$.

The idea is to replace the strings by their shortest programs. For Kolmogorov complexity with unbounded resources, a string $x$ is ``interchangeable'' with one of its shortest programs $p$ in the following sense: 
\[
\KS(x\cnd p) = O(\log m) \quad \text{and}\quad \KS(p\cnd x) = O(\log m),
\]
where $m = \KS(x)$.
The first part is obvious for \emph{every} program $p$ (even with $O(1)$ instead of $O(\log m)$): we apply the optimal interpreter to~$p$ and obtain $x$. The second part is also pretty simple: given  $x$ and  $m$, we run the optimal interpreter on all programs of length $m$ in parallel and take the first one that produces~$x$.

A similar property is true for the space bounded Kolmogorov complexity. If $\KS^s(x)=m$, then there is a program $p$ of length $m$ such that
\[
\KS^{s+O(1)}(x\cnd p) = O(\log m)\quad \text{and}\quad  \KS^{s+O(|x|)}(p\cnd x) = O(\log m).
\]
This $p$ is one of the programs of length $m$ that produce $x$ in space $s$. For such a program~$p$ the first part is trivial: we simulate the universal interpreter on~$p$ and obtain $x$, with $O(1)$ space overhead for the simulation. For the second part, we show how to find a program $p$ of length $m$ for $x$ (that works in space $s$) given~$x$ and~$m$. Following the argument we already used, we try all the programs of length $m$ giving them more and more space ($s'=1,2,3,\ldots$) until one of them produces $x$. For keeping the space overhead in this process small, we keep the value of the current space bound $s'$ indirectly, as a difference between $s'$ and current length of the stacks. We need also to keep $x$, thus the $O(|x|)$ overhead in the space bound.

Using this property, we replace $x$ and $y$ by some $p_x$ and $p_y$ whose lengths are $O(\KS^s(x,y))$. A small technicality is that $\KS^s(x)$ may not be bounded by $\KS^s(x,y)$, since extracting $x$ from the pair requires some overhead. In fact, $O(1)$ overhead is enough: $\KS^{s+O(1)}(x)\le \KS^{s}(x,y)$ (but $O(|x|+|y|)$ overhead would work too). Applying the previous remark to this bound, we find $p_x$ of length at most $\KS^{s}(x,y)$ such that  
\[
\KS^{s+O(1)}(x\cnd p_x) = O(\log\KS^s(x,y))\quad \text{and}\quad  \KS^{s+O(|x|)}(p_x\cnd x) = O(\log\KS^s(x,y)).
\]
The same can be done for $y$ to get a replacement string $p_y$ with similar properties. Then we apply the previous form of the inequality (with $O(\log n)$) to $p_x$ and $p_y$, and note that replacing $x$ by $p_x$ in expressions with conditional or unconditional complexity changes the complexity bound by $O(\log\KS^s(x,y))$ and the overhead by $O(|x|+|y|)$. 

This finishes the proof of Theorem~\ref{th:sym}.
\end{proof}

\section{Basic inequalities: space-bounded version}\label{sec:basic}

We have defined space-bounded complexity for pairs. In the same way (and with the same precision) one can define the complexity of triples, and, in general, $m$-tuples for every fixed $m$. In the section we prove space-bounded versions of the so-called \emph{basic inequalities} for Kolmogorov complexity.

The basic inequality involves complexities of triples and says (in the unbounded version) that 
\[
\KS(x)+\KS(x,y,z)\le \KS(x,y)+\KS(x,z)+ O(\log n)
\]
if $x,y,z$ are strings of length at most $n$. Usually it is proved by considering conditional complexities:
\begin{align*}
\KS(x,y)&=\KS(x)+\KS(y\cnd x)+O(\log n),\\
\KS(x,z)&=\KS(x)+\KS(z\cnd x)+O(\log n),\\
\KS(x,y,z)&=\KS(x)+\KS(y,z\cnd x)+O(\log n).
\end{align*}
Using these equalities, we rewrite the inequality as
\[
\KS(y,z\cnd x)\le \KS(y\cnd x)+\KS(z\cnd x)+O(\log n),
\]
and this is a relativized version of the inequality for the complexity of pairs:
\[
\KS(y,z) = \KS(y)+\KS(z\cnd y) +O(\log n) \le \KS(y)+\KS(z)+O(\log n);
\]
adding $y$ as a condition may only decrease the complexity of $z$. In computability theory relativization is usually understood as adding an oracle access to some set to all the computations; almost all results of general computability theory remain valid after relativization.  In algorithmic information theory a slightly different notion of relativization is also used: instead of adding a set as an oracle, we add some string as a condition in all the complexity expressions. Almost all results (and their proofs) remain valid after that.\footnote{For the space-bounded complexity additional precautions are needed: if we add $x$ as a condition, it may be necessary to add $O(|x|)$ to the space bound.}

 Let us do this in more detail to see how the space-bounded version can be proven. We have
\begin{align*}
\KS^{s'}(x)+\KS^{s'}(y\cnd x)& \le \KS^s(x,y)+O(\log n),\\
\KS^{s'}(x)+\KS^{s'}(z\cnd x)& \le \KS^s(x,z)+O(\log n),
\end{align*}
for some $s'$ slightly larger than $s$ (by that we mean that $s'=s+O(n)$). Therefore,
\[
2\KS^{s'}(x) + \KS^{s'}(y\cnd x) + \KS^{s'}(z\cnd x) \le \KS^s(x,y)+\KS^s(x,z)+O(\log n).
\]
From this we conclude that
\[
2\KS^{s'}(x) + \KS^{s''}(y,z\cnd x) \le 
2\KS^{s'}(x) + \KS^{s'}(y\cnd x) + \KS^{s'}(z\cnd x)+O(\log n) \le
\KS^s(x,y)+\KS^s(x,z)+O(\log n),
\]
for some $s''$ slightly exceeding $s'$, using the relativized inequality for the complexity of a pair:
\[
\KS^{s''}(y,z\cnd x) \le \KS^{s'}(y\cnd x) + \KS^{s'}(z\cnd x)+O(\log n).
\]
Here $s''=s'+O(n)=s+O(n)$ absorbs the increase $O(n)$ caused by the length of the condition $x$ that is needed for the relativization. Now we recall that
\[
\KS^{s'''}(x,y,z) \le \KS^{s'}(x)+\KS^{s''}(y,z\cnd x)+O(\log n)
\] 
(the easy direction of the Kolmogorov--Levin formula) for some $s'''$ slightly greater than $s'$ and $s''$, and get
\[
\KS^{s'}(x)+\KS^{s'''}(x,y,z) \le \KS^s(x,y)+\KS^s(x,z)+O(\log n).
\]
For uniformity we can replace $s'$ by $s'''$ on the left-hand side.
Here $s'''$ is the third iteration of adding overhead, so still 
\(
s''' = s+O(n),
\)
and we get the following space-bounded version of basic inequality:
\begin{theorem}[Space-bounded basic inequality]\label{th:basic}
\[
\KS^{s'} (x) + \KS^{s'} (x,y,z)\le
\KS^s(x,y)+\KS^s(x,z)+O(\log n) 
\]
for all~$n$, for all strings $x,y,z$ of length at most~$n$, for all numbers~$s$, and for 
\(
  s'=s+O(n).
\)
\end{theorem}
More general inequalities (called also basic inequalities) are obtained if we replace $x,y,z$ by tuples of strings; they are easy corollaries of Theorem~\ref{th:basic} (converting the tuples into their string encoding and vice versa can be done in $O(n)$ space for strings of size at most $n$).

\section{Shannon inequalities: iterations}\label{sec:shannon}

Fix some integer $k\ge1$; let $x_1,\ldots,x_k$ be some strings. For each $I\subset\{1,\ldots,k\}$ we consider the tuple $x_I$ made of strings $x_i$ with $i\in I$. In this notation, the basic inequalities mentioned above can be written as
\[
\KS^{s'}(x_{I\cap J})+\KS^{s'}(x_{I\cup J})\le \KS^s(x_I)+\KS^s(x_J)+O(\log n),
\]
if all $x_1,\ldots,x_k$ are strings of length at most $n$ and $s'=s+O(n)$. (The constants in the $O$-notation may depend on $k, I, J$, but not on $n$, $x_1,\ldots,x_k$, $s$.)

Taking the sum of several basic inequalities (for the same $k$, but for different $I$ and $J$), we may get other linear inequalities for the complexities of tuples, i.e., inequalities of the type
\[
\sum_{I\subset \{1,\ldots,k\}} \lambda_I\KS(x_I)\ge 0,
\]
where $\lambda_I$ are some real coefficients. This is a well known procedure for unbounded Kolmogorov complexity~\cite[Chapter 10]{usv}; the resulting linear inequalities are called \emph{Shannon inequalities}. Not all linear inequalities that are true with logarithmic precision are Shannon inequalities (an important discovery made in~\cite{yeung-zhang}). 

In this section we show that \emph{every Shannon inequality has a space-bounded version}. This space-bounded version is constructed as follows. We start by separating the positive and negative coefficients in the linear inequality. The corresponding groups are denoted by $L$ and $R$; their elements are subsets of the set $\{1,\ldots,k\}$, and we assume that $L\cap R = \varnothing$. Now the general form of a linear inequality for complexities of tuples is
\begin{equation}
\sum_{I\in L} \lambda_I\KS(x_I)\le \sum_{J\in R} \mu_J\KS(x_J)\label{eq:separ}
\end{equation}
where all $\lambda_I$ and $\mu_J$ are non-negative. The following theorem says that each Shannon inequality has a space-bounded counterpart of the same form as for the basic inequalities (but with slightly weaker space bound).

\begin{theorem}\label{th:shannon}
Consider a linear inequality of the form~\textup{(\ref{eq:separ})} that is a linear combination of basic inequalities \textup(is a Shannon inequality\textup). Then the following space-bounded version of this inequality is true:
\begin{equation}\label{eq:shannon}
\sum_{I\in L} \lambda_I\KS^{s'}(x_I)\le \sum_{J\in R} \mu_J\KS^s(x_J)+O(\log n),
\end{equation}
if $x_1,\ldots,x_k$ are strings of length at most $n$, and $s'=s+O(n^2)$. 
\end{theorem}

Here the constants in the $O$-notation depend on the inequality (more precisely, on~$k$ and the coefficients $\lambda_I$ and~$\mu_j$)\footnote{The constant in the last line depends only on $k$, as the proof shows.}, but neither on $n$ nor on $x_1,\ldots,x_k$. Note that the overhead is worse than for the basic inequalities: we have $O(n^2)$ instead  of $O(n)$.

\begin{proof}
Consider the basic inequalities whose sum is the inequality (\ref{eq:separ}). For each of them consider the space-bounded version (from Theorem~\ref{th:basic}). The sum of these space-bounded inequalities does not give $(\ref{eq:shannon})$ directly: the resulting inequality may have terms $\KS(x_I)$ with the same $I$ in the left and right hand sides. In other words, we get an inequality of type (\ref{eq:separ}), but the sets $L$ and $R$ are not necessarily disjoint. For the unbounded complexities, these terms just cancel each other (partially or completely), and we get the desired Shannon inequality. Now, when adding the space-bounded versions of the same basic inequality, we get an inequality where the complexity of the same tuple may appear with the same coefficient on both sides, but with different space bounds. We can rewrite is as
\begin{equation*}
\sum_{I\in L} \lambda_I\KS^{s'}(x_I)+\sum_{K\in C}\sigma_K \KS^{s'}(x_K)\le \sum_{J\in R} \mu_J\KS^s(x_J)+\sum_{K\in C}\sigma_K \KS^{s}(x_K)+O(\log n).
\end{equation*}
Here $x_K$ are tuples that appear on both sides in the terms that are canceled in the unbounded version (partially or completely). Some $K\in C$ may also appear in $L$ or $R$ (the part that is not canceled), but not in both: the sets $L$ and $R$ are disjoint. We would like to cancel the complexities of $x_K$ for $K\in C$, but now the complexities are different. They have bound $s'$ on the left-hand side and $s$ on the right-hand side, and cannot be canceled anymore.

The following trick helps. Let $f(s)=s+O(n)$ be the function from Theorem~\ref{th:basic} that transforms the right-hand side bound  $s$ to the left-hand side bound $s'$ (here $s$ is a variable, while $n$ and the constants in the $O(\cdot)$-notation are fixed). Consider the sequence of space bounds
\[
 u_0 = s, u_1=f(u_0),\ldots, u_N=f(u_{N-1})
\]
for some large $N$. All tuple complexities can only decrease if we increase the space bound from $u_t$ to $u_{t+1}$. Therefore, for a large enough $N$, namely, $N=O(n)$ with a large enough constant, we guarantee the existence of $t$ such that all complexities of tuples are the same with bounds $u_t$ and $u_{t+1}$. Then we can add the space-bounded inequalities and cancel the common terms as we did for the unbounded version. More precisely, we know that
\begin{equation*}
\sum_{I\in L} \lambda_I\KS^{u_{t+1}}(x_I)+\sum_{K\in C}\sigma_K \KS^{u_{t+1}}(x_K)\le \sum_{J\in R} \mu_J\KS^{u_t}(x_J)+\sum_{K\in C}\sigma_K \KS^{u_t}(x_K)+O(\log n),
\end{equation*}
and on both sides $u_t$ can be replaced by $u_{t+1}$ due to our assumption. So we can cancel the common terms. We cannot compute $t$ for which there is no change in the complexities, but its existence is guaranteed. Then we can replace the bound on the left-hand side by $u_N$, and on the right-hand side by $s=u_0$.

  It remains to note that for $N=O(n)$ we add the $O(n)$ term $N=O(n)$ times, so the final value after $N$ iterations is $s+O(n^2)$. 

In fact, we can improve the bound in Theorem~\ref{th:shannon}. For that we may note that it is not needed to have exactly the same complexities with space bounds $u_{t+1}$ and $u_t$. It is enough that the difference between them is $O(\log n)$, since we have $O(\log n)$ term in the right hand side anyway. Therefore, $N=n/\log n$ iterations are enough, and in this way we replace $O(n^2)$ by $O(n^2/\log n)$, getting a bit stronger version of Theorem~\ref{th:shannon}.
\end{proof}

\begin{remark}
This argument relies on the good space bounds in the left hand side of Theorem~\ref{th:sym}. If we used (instead of Theorem~\ref{th:sym}) the bound with factor~$2$, the $n$-th iteration would give an exponential factor $2^n$, so we wouldn't get a polynomial (in $n$) space bound.\footnote{In the previous version of this paper (still available in \texttt{arxiv},~\cite{this-archive}) we had $f(s)=s+O(\log s)+O(n)$, and then we estimated the iterations of $f$ by a simple but boring argument. With a better bound $f(s)=s+O(n)$ this is no more needed.} 
\end{remark}

\begin{remark}
It may happen that for some Shannon inequality the cancellation problem does not arise. This indeed happens for some natural Shannon inequalities, e.g., for
\[
2\KS^{s'}(A,B,C)\le \KS^{s}(A,B)+\KS^{s}(A,C)+\KS^{s}(B,C)+O(\log n),
\]
 that is therefore true for $s'=s+O(n)$. However, it is not clear whether this can be done for arbitrary Shannon inequalities. 
\end{remark}

\section{General result}\label{sec:main}

In this section we use a similar technique to prove a more general result that covers not only Shannon inequalities but all true linear inequalities for Kolmogorov complexity. Recall that a theorem from Hammer et al. (\cite{romashchenko1997}, see~\cite[Chapter 10]{usv} for the detailed exposition) says that the same linear inequalities are true for complexities (with logarithmic precision) and for Shannon entropies. In this section we want to show that all inequalities in this class have space-bounded counterparts.  For that, we need to modify the original proof from~\cite{romashchenko1997,usv} using the tools we developed. Let us first formulate this result in a form similar to Theorem~\ref{th:shannon}.
\medskip

\noindent
Let us fix some integer~$k\ge 1$. 

\begin{theorem}[Inequality with two space bounds]\label{th:general}
Assume that a linear inequality for unbounded complexities with non-negative coefficients $\lambda_I$ and $\mu_J$,
\begin{equation}
\sum_{I\in L} \lambda_I\KS(x_I)\le \sum_{J\in R} \mu_J\KS(x_J)+O(\log n),
\end{equation}
is true for all $n$ and for all strings $x_1,\ldots,x_k$ of length at most $n$.
Then its space-bounded version
\begin{equation}
\sum_{I\in L} \lambda_I\KS^{s'}(x_I)\le \sum_{J\in R} \mu_J\KS^s(x_J)+O(\log n)\label{ineq:bound}
\end{equation}
holds for all $n, s$, for all strings $x_1,\ldots,x_k$ of length at most $n$ and for $s'=s+O(n^2)$.
\end{theorem}

We will derive this result from a different statement that does not require separating positive and negative coefficients:

\begin{theorem}[Existence of a common space bound]\label{th:general2}
Assume that a linear inequality for unbounded complexities
\begin{equation}
\sum_{I} \lambda_I\KS(x_I)\le O(\log n),
\end{equation}
with coefficients \( \lambda_{I} \) that can be positive or negative,
is true for all $n$ and for all strings $x_1,\ldots,x_k$ of length at most $n$.
Then for every $n,s$ and for every $x_1,\ldots,x_k$ of length at most $n$ its space-bounded version
\begin{equation}
\sum_{I\in L} \lambda_I\KS^{s'}(x_I)\le O(\log n)\label{ineq:exist}
\end{equation}
holds for some $s'\in [s, s+O(n^2)]$.
\end{theorem}

This statement is purely existential (and a bit weird): it says that \emph{there exists some $s'$} between $s$ and $s+O(n^2)$ (depending on $x_1,\ldots,x_k$) for which the inequality is true. Still it is easy to see that Theorem~\ref{th:general} immediately follows from Theorem~\ref{th:general2}: if the inequality~(\ref{ineq:exist}) is true for some value of $s'$, we may separate positive and negative coefficients as in (\ref{ineq:bound}) and then replace $s'$ by $s$ in the right hand side, and by $s+O(n^2)$ in the left hand side, due to the monotonicity. So it remains to prove Theorem~\ref{th:general2}.

\begin{proof}
We adapt the arguments used in~\cite{romashchenko1997,usv} to prove the connection between (unbounded) complexity and entropy inequalities.

\emph{Step 1}. First of all, we convert our assumption into the language of Shannon's information theory and note that
\begin{equation*}
\sum_{I} \lambda_I H(\xi_I)\le 0
\end{equation*}
for arbitrary random variables $\xi_1,\ldots,\xi_k$. Indeed, it is well known (see, e.g., \cite[chapter 7]{usv}) that if $\rho$ is an arbitrary random variable that has finite range, and $\rho^1,\ldots,\rho^N$ are independent identically distributed copies of $\rho$, then the expected Kolmogorov complexity of the finite object $(\rho^1,\ldots,\rho^N)$ is $NH(\rho)+O(\log N)$. Then, for a large $N$, we take $N$ independent copies of the tuple $\xi_1,\ldots,\xi_k$. For every $I\subset\{1,\ldots,k\}$ we have
\[
H(\xi_I)= \frac{\E[\KS(\xi_I^1,\ldots,\xi_I^N)]}{N} + \frac{O(\log N)}{N}.
\]
The matrix $\xi_i^j$ can be considered as a $k$-tuple of its columns (for each column $i$ is fixed and $j$ ranges from $1$ to $N$), and the inequality for complexities can be applied to these columns. It guarantees that
\[
\sum_I \lambda_I H(\xi_I) \le \frac{O(\log N)}{N},
\]
and we get the required inequality since the left hand side does not depend on $N$.

\emph{Step 2}. For a given tuple $x_1,\ldots,x_k$ whose elements are strings of length at most~$n$, and for some $s'\ge s$ consider the set $X$ of all the tuples $y_1,\ldots,y_k$ of strings of length at most~$n$ such that
\[
\KS^{s'}(y_I \cnd y_J) \le \KS^{s'}(x_I\cnd x_J)
\]
	for all sets $I,J\subset\{1,\ldots,k\}$. The log-size of $X$ does not exceed $\KS^{s'}(x_1,\ldots,x_k)$, since one of the inequalities requires that $\KS^{s'}(y_1,\ldots,y_k)\le \KS^{s'}(x_1,\ldots,x_k)$ (for empty $J$ and for $I=\{1,\ldots,k\}$). The following lemma provides a lower bound for its size:

\begin{lemma}\label{lem:typization}
The log-size of $X$ is at least $\KS^{s''}(x_1,\ldots,x_k)-O(\log n)$, where $s''=s'+O(n)$.
\end{lemma}
	
\begin{proof}[Proof of Lemma~\ref{lem:typization}]
The set of all $y_1,\ldots,y_k$ of length at most~$n$ that satisfy all the inequalities
\[
\KS(y_I \cnd y_J) \le \KS^{s'}(x_I\cnd x_J)
\]
(with unbounded complexity in the left side) can be enumerated if $n$ and all the complexities in the right side on the inequalities are given. So the information needed to start the enumeration is of size $O(\log n)$. The tuple $x_1,\ldots,x_k$ belongs to $X$, and can be reconstructed if its ordinal number in the enumeration is given. Therefore,
\[ 
\KS (x_1,\ldots,x_k) \le \log|X|+O(\log n).
\]
Let us strengthen this inequality by using bounded complexity in the left-hand side:
\[
\KS^{s'+O(n)} (x_1,\ldots,x_k)\le \log |X| + O(\log n).
\]
  Indeed, the enumeration can be performed sequentially with increasing space bounds $1,2,3,\ldots$, using Lemma~\ref{lem:strongcomp} to compute space-bounded complexities. As before, we ensure the enumeration without repetitions by checking for every tuple $y_1,\ldots,y_k$ whether it already appeared for the previous value of the space bound. In this enumeration the tuple $x_1,\ldots,x_k$ appears when the space bound is $s'$ (or less). Stopping the enumeration at this time (knowing the number of tuples that should be enumerated), we use space $s''=s'+O(n)$. As in the proof of Theorem~\ref{th:sym}, we keep the current value of the space bound all the time, but in such a way (as the difference between this value and stacks' size) that the used space never exceeds~$s'+O(n)$.  

Lemma~\ref{lem:typization} is proven.
\end{proof}

\emph{Step 3}. As in the proof of Theorem~\ref{th:shannon}, consider the sequence of bounds $s, f(s), f(f(s)),\ldots$ where $f(s)=s+O(n)$ is the bound from Lemma~\ref{lem:typization}. When the bound $s'$ increases, all the complexities $\KS^{s'}(x_I\cnd x_J)$ may only decrease. Recall that the parameter $k$ is fixed; we have only $O(1)$-many decreasing complexities (for all pairs $I,J\subset\{1,2,\ldots,k\}$), and the initial value of these complexities is $O(n)$. Therefore, there are at most $O(n)$ steps when some complexity decreases, and it is enough to make $O(n)$ iterations to come to an iteration step when all complexities do not change. The total increase of the space bound during $O(n)$ iterations is $O(n^2)$. So we come to the following statement:
\begin{lemma}\label{lem:stabilization}
There exists some $s'\in [s, s+O(n^2)]$ such that 
\[
\KS^{f(s')}(x_I\cnd x_J)=\KS^{s'}(x_I\cnd x_J)
\]
for all $I,J\subset\{1,2,\ldots,k\}$.
\end{lemma}

Combining Lemma~\ref{lem:stabilization} with Lemma~\ref{lem:typization}, we conclude that there exists $s'\in [s, s+O(n^2)]$ such that the set $X$ of all $y_1,\ldots,y_k$ of length at most~$n$ such that
\[
\KS^{s'}(y_I \cnd y_J) \le \KS^{s'}(x_I\cnd x_J)
\]
(for all $I$, $J$) has log-size $\KS^{s'}(x_1,\ldots,x_k)+O(\log n)$: the upper bound for $\log|X|$ is obvious, and the lower bound is provided by Lemma~\ref{lem:typization}, where we can replace $s''$ by $s'$ due to the choice of $s'$.

Now consider a tuple of random variables $\xi_1,\ldots,\xi_k$ uniformly distributed in the set $X$. Its entropy is $\log|X|=\KS^{s'}(x_1,\ldots,x_k)+O(\log n)$. The following lemma shows that the same connection between entropies and complexities is true for an arbitrary subset of indices. By $\xi_I$ we denote the tuple of random variables $\xi_i$ for $i\in I$. 

\begin{lemma}\label{lem:subtuples}
\[
H(\xi_I) = \log \KS^{s'}(x_I)+O(\log n).
\]
for every $I\subset\{1,\ldots,k\}$.
\end{lemma}
This lemma finishes the proof of Theorem~\ref{th:general2}. Indeed, if some inequality is true for (unbounded) complexities with logarithmic precision, it is true for entropies. In particular, it is true for entropies of subsets of $\xi_1,\ldots,\xi_k$, and these entropies coincide with bounded-space complexities of corresponding subsets of $x_1,\ldots,x_k$ with logarithmic precision. Therefore the inequality is also true for bounded-space complexities (for some $s'$ in the interval $[s,s+O(n^2)]$, provided by Lemma~\ref{lem:stabilization}). It remains to prove Lemma~\ref{lem:subtuples}.

\begin{proof}[Proof of Lemma~\ref{lem:subtuples}]
Let $I$ be some subset of $\{1,\ldots,k\}$, and $J$ be its complement: $J=\{1,\ldots,k\}\setminus I$. We know that
\[
H(\xi_1,\ldots,\xi_k)= H(\xi_I)+H(\xi_J\cnd \xi_I).
\]
All values of $\xi_I$ are among tuples $y_I$ for $y\in X$, and therefore $\KS^{s'}(y_I)\le \KS^{s'}(x_I)$. So the range of $\xi_I$ has log-size at most $\KS^{s'}(x_I)+O(1)$, and the entropy of a random variable does not exceed the log-size of its range:
\[
H(\xi_I)\le \KS^{s'}(x_I)+O(1).
\]
For similar reasons we have
\[
H(\xi_J\cnd \xi_I)\le \KS^{s'}(x_J\cnd x_I)+O(1).
\]
Indeed, for every $y_1,\ldots,y_k$ in $X$ we have $\KS^{s'}(y_J\cnd y_I)\le \KS^{s'}(x_J\cnd x_I)$, so for every value of $\xi_I$ the set of possible values of $\xi_J$ has log-size at most $\KS^{s'}(x_J\cnd x_I)+O(1)$.
The choice of $s'$ guarantees that the complexities of $x_J$ given $x_I$ with bound $s'$ coincide with the same complexities with bound $s''=s'+O(n)$. So we can write a chain of inequalities with precision $O(\log n)$:
\begin{equation*}
H(\xi_1,\ldots,\xi_k)= H(\xi_I)+H(\xi_J\cnd \xi_I)\le \KS^{s'}(x_I)+\KS^{s'}(x_J\cnd x_I) 
= \KS^{s''}(x_I)+\KS^{s''}(x_J\cnd x_I) \le \KS^{s'}(x_1,\ldots,x_k).
\end{equation*}
(the last inequality is due to Theorem~\ref{th:sym}).
We know that the leftmost and rightmost terms of this inequality coincide (with $O(\log n)$ precision, as for the other parts), so all the inequalities that appear in this chain are equalities with $O(\log n)$ precision. In particular, $H(\xi_I)=\KS^{s'}(x_I)+O(\log n)$. Lemma~\ref{lem:subtuples} is proven.
\end{proof}

This finishes the proof of Theorem~\ref{th:general2} (and its corollary, Theorem~\ref{th:general}).	
\end{proof}

\begin{remark}
Again, we do not need the complexities in Lemma~\ref{lem:stabilization} with bounds $s$ and $s'$ to be \emph{exactly} the same; all the arguments remain valid if we make them differ by $O(\log n)$. In this way we may use $O(n/\log n)$ steps instead of $O(n)$, and get a slightly better bound $O(s)+O(n^2/\log n)$ in Theorems~\ref{th:general} and \ref{th:general2}.
\end{remark}

\begin{remark}
We may consider a more general class of linear inequalities in Theorem~\ref{th:general2} that include all conditional complexities:
\[
\sum \lambda_{I,J} \KS(x_I\cnd x_J)\le 0.
\]
Theorem~\ref{th:general2} remains valid, and the proof is essentially the same; we need to show in Lemma~\ref{lem:subtuples} that
\[
H(\xi_I\cnd \xi_J) = \log \KS^{s'}(x_I\cnd x_J)+O(\log n)
\]
for all $I,J\subset\{1,\ldots,k\}$. This is done by a similar argument. First let us assume that $I$ and $J$ are disjoint. Let $R$ be the set of indices that are not in $I$ and not in $J$. Then we write the following chain of inequalities with $O(\log n)$ precision:
\begin{multline*}
H(\xi_1,\ldots,\xi_k)=H(\xi_J)+H(\xi_I\cnd \xi_J)+H(\xi_R\cnd \xi_{I\cup J})\le
\KS^{s'}(x_J)+\KS^{s'}(x_I\cnd x_J)+\KS^{s'}(x_R\cnd x_{I\cup J})=\\=
\KS^{s''}(x_J)+\KS^{s''}(x_I\cnd x_J)+\KS^{s''}(x_R\cnd x_{I\cup J}) \le
\KS^{s'}(x_1,\ldots,x_k),
\end{multline*}
and use the same argument as before. The difference is that here we need to use the bounded-space Kolmogorov--Levin formula for triples:
\[
\KS^{s''}(x)+\KS^{s''}(y\cnd x)+\KS^{s''}(z\cnd x,y)\le \KS^{s'}(x,y,z)
\]
which can be obtained by using the formula for pairs twice; recall that $O(n)$ overhead, appearing twice, is still $O(n)$.

As before, Theorem~\ref{th:general2} implies Theorem~\ref{th:general}. 

For unbounded complexities it makes no sense to include conditional complexities in the inequalities, since Kolmogorov--Levin formula reduces them to unconditional ones. However, for space-bounded complexities this  reduction will change the bounds, so we may wish to allow them to appear explicitly.
\end{remark}

\begin{remark}

\leavevmode

In Theorems~\ref{th:general} and~\ref{th:general2} we may also replace the $O(\log n)$ additive term by $O(\log \KS^s(x_1, \ldots, x_k))$. 

For Theorem~\ref{th:general} we repeat the argument used to finish the proof of theorem~\ref{th:sym}. We noted there that for all~$s$ and~$x$, there exists a program $p$ of length $m=\KS^s(x)$ such that $\KS^{s+O(|x|)}(p \cnd x) \le O(\log m)$ and $\KS^{s+O(1)}(x\cnd p) \le O(1)$. Hence, the better precision follows by replacing $x_1, \ldots, x_k$ by the programs $p_1, \ldots, p_k$ of length $|p_i| \le \KS^{s+c}(x_i)$ where the constant~$c$ should be large enough to guarantee that $\KS^{s+c}(x_i) \le \KS^{s}(x_1, \ldots, x_k)+ O(1)$. 

For Theorem~\ref{th:general2} we use the same idea, but first we have to look at the proof of this theorem and notice that in fact we proved the following statement: 
\begin{quote}
\emph{if for some strings $x_1,\ldots,x_k$ of length at most $n$ the complexities $\KS^s(x_I\cnd x_J)$ change at most by $d$ when $s$ is increased up to $s+cn$} (here $c$ is a large enough constant), \emph{then the inequality $(\ref{ineq:exist})$ is valid for $s'=s$ with additional term $O(d)$ in the right hand size.} 
\end{quote}

Now the argument goes as follows. We have strings $x_1,\ldots,x_k$ of length at most $n$ and look at the complexities $\KS^{s'}(x_I\cnd x_J)$ as a function of $s'$. As before, we can find a interval of length $c'n$ inside $[s,s+O(n^2)]$ where all these complexities do not change. This can be done for arbitrary large constant $c'$ (and the constant in $O(n^2)$ depends on $c'$). Let $[u,v]$ be this interval. Then we have $\KS^u(x_i)\le \KS^u(x_1,\ldots,x_k)+O(1)$, since $\KS^{s'}(x_i)$ is the same for all $s'\in[u,v]$. 

  Now we apply our replacement argument and find $p_i$ such that conditional complexities $\KS^{u+O(n)}(x_i\cnd p_i)$ and $\KS^{u+O(n)}(p_i\cnd x_i)$ are at most $O(\log m)$, and the lengths of all $p_i$ are $O(m)$, where $m=O(\KS^u(x_1,\ldots,x_k))$. Therefore, if we increase the left endpoint $u$ of the interval for the space bounds by $O(n)$,  in this smaller interval $[u+O(n),v]$ all the values $\KS^{s'}(p_I\cnd p_J)$ differ from corresponding $\KS^{s'}(x_I\cnd x_J)$ by at most $O(\log m)$ and therefore change (when $s'$ is in $[u+O(n),v]$) at most by $O(\log m)$, since $\KS^{s'}(x_I\cnd x_J)$ do not change at all. It remains to apply the result quoted earlier to $p_1,\ldots, p_k$. Note that the lengths of $p_1,\ldots,p_k$ are $O(m)$, that $\KS^{s'}(x_I\cnd x_J)$ are $O(\log m)$ close to $\KS^{s'}(p_I\cnd p_J)$, and that the remaining interval $[u+O(n),v]$ has length at least $cn$ for any constant $c$ if $c'$ is large enough.
\end{remark}

\section{Discussion}\label{sec:discussion}

\subparagraph*{Increasing the density.} Theorem~\ref{th:general2} says that the space-bounded version of the inequality (that is true in the unbounded version) is valid for the sequence of space bounds $s_j$ that is not very sparse: $s_{j+1}\le s_j+O(n^2)$.  Is it possible to improve this result and show that the inequality in question is true for ``more dense'' sequence of space bounds?

\subparagraph*{Space-bounded versions of other results.} Our results are part of the space-bounded version of algorithmic information theory. In general, one could take some notion or theorem of algorithmic information theory and look for its space-bounded counterpart. For Muchnik's conditional codes theorem this was done by Musatov (see~\cite{musatov2014} and references therein). 

However, there are many problems in this approach. For example, if we define  mutual information with space bound $s$ in a natural way as
\[
I^s(a:b) = \KS^s(a)- \KS^s(a\cnd b),
\]
this notion is not monotone; a priori the mutual information can oscillate when $s$ increases. It would be interesting to understand what kinds of oscillations are possible. Is it possible that two strings are mutually independent for some space bound, then dependent for some larger bound, then again independent, and so on? Also the relations between $I^s(a:b)$, $I^s(b:a)$ and the symmetric expression $\KS^s(a)+\KS^s(b)-\KS^s(a,b)$ are unclear.

\subparagraph*{Time-bounded versions.} We can try a similar approach for time bounds (instead of space bounds). It also works, but the natural bound in the formula for complexity of pairs multiplies the time complexity by $2^{O(n)}$; also the simulation would increase time significantly (for a one-tape machine the simulation of $t$ steps needs more than $t^2$ time). When we iterate these bounds $O(n)$ times, we get ridiculously high time bounds. It is just good luck that Sipser's trick for space bounds allows us to get some reasonable space bounds, and for time bounds things are much worse. Still one can have some versions of our results with computable (though ridiculously large) time bounds.

\end{document}